\documentclass[11pt]{article}
 
\usepackage[latin1]{inputenc}    
\usepackage[OT1]{fontenc}
\usepackage[english]{babel}
\usepackage{fullpage}
\usepackage{amsmath}
\usepackage{amssymb}
\usepackage{amsthm}
\usepackage{mathrsfs}
\usepackage{graphicx}
\usepackage{dsfont}
\usepackage{wrapfig}
\usepackage{comment}
\usepackage{rotating}
\usepackage{hyperref} 
\usepackage[all]{xy}
\usepackage{color}

\newcommand{\R}{\mathbb{R}}
\newcommand{\Rop}{\mathbb{R}^{\text{op}}}
\newcommand{\Rext}{\mathbb{R}_{\text{Ext}}}
\newcommand{\N}{\mathbb{N}}
\newcommand{\Z}{\mathbb{Z}}

\newcommand{\Reeb}{\mathrm{R}}
\newcommand{\Ord}{\mathrm{Ord}}
\newcommand{\Ext}{\mathrm{Ext}}
\newcommand{\Rel}{\mathrm{Rel}}
\newcommand{\Dg}{\mathrm{Dg}}

\newcommand{\im}{\mathrm{im}}

\newcommand{\Merge}{\mathrm{Merge}}

\newcommand{\Crit}{\mathrm{Crit}}

\newcommand{\distfd}{d_{\rm FD}}
\newcommand{\distb}{d_{\rm B}}
\newcommand{\hdistfd}{\hat{d}_{\rm FD}}
\newcommand{\hdistb}{\hat{d}_{\rm B}}

\newcommand{\e}{\varepsilon}

\newcommand{\const}{22}
\newcommand{\RS}{{\sf Reeb}}

\newtheorem{thm}{Theorem}[section]
\newtheorem{corollary}[thm]{Corollary}
\newtheorem{proposition}[thm]{Proposition}
\newtheorem{lemma}[thm]{Lemma}
\newtheorem{remark}[thm]{Remark}
\newtheorem{definition}[thm]{Definition}

\title{Local Equivalence and Intrinsic Metrics between Reeb Graphs
\footnote{This work was partially supported by ERC Grant
Agreement No. 339025 GUDHI (Algorithmic Foundations of Geometry Understanding in Higher Dimensions) and was carried out while the second author was visiting the ICERM at Brown University.}}
\author{Mathieu Carri\`ere, Steve Oudot}

\begin{document}

\maketitle

\begin{abstract}
As graphical summaries for topological spaces and maps, Reeb graphs
are common objects in the computer graphics or topological data
analysis literature.  Defining good metrics between these objects has
become an important question for applications, where it matters to
quantify the extent by which two given Reeb graphs differ.  Recent
contributions emphasize this aspect, proposing novel distances such as
{\em functional distortion} or {\em interleaving} that are provably
more discriminative than the so-called {\em bottleneck distance},
being true metrics whereas the latter is only a pseudo-metric. Their
main drawback compared to the bottleneck distance is to be
comparatively hard (if at all possible) to evaluate. Here we take the
opposite view on the problem and show that the bottleneck distance is
in fact good enough {\em locally}, in the sense that it is able to
discriminate a Reeb graph from any other Reeb graph in a small enough
neighborhood, as efficiently as the other metrics do.
This suggests considering the
{\em intrinsic metrics} induced by these distances, which turn out to
be all {\em globally} equivalent. 
This novel viewpoint on the study of Reeb graphs has a potential impact on
applications, where one may not only be interested in discriminating
between data but also in interpolating between them.
\end{abstract}

\section{Introduction}

In the context of shape analysis, the Reeb graph~\cite{Reeb46} provides a meaningful summary
of a topological space and a real-valued function defined on that space.
Intuitively, it continuously collapses the connected components of the level sets of the function into single points, thus tracking
the values of the functions at which the connected components merge or split.
Reeb graphs have been widely used in computer graphics and visualization---see~\cite{Biasotti08} for a survey, 
and their discrete versions, including the so-called {\em Mappers}~\cite{Singh07}, have become emblematic tools
of topological data analysis due to their success in applications~\cite{NBA, Barra14, Lum13, Nicolau11}. 

Finding relevant dissimilarity measures for comparing Reeb graphs has become an important question in the recent years.
The quality of a dissimilarity measure is usually assessed through three criteria: 
its ability to satisfy the axioms of a metric, its discriminative power, and its computational efficiency.  	
The most natural choice to begin with is to use the {\em Gromov-Hausdorff distance} $d_{\rm GH}$~\cite{Burago01} for Reeb graphs 
seen as metric spaces.  
The main drawback of this distance is to quickly become intractable to compute in practice, even for graphs that are metric trees~\cite{Agarwal15}.
Among recent contributions, the {\em functional distortion distance}
$\distfd$~\cite{Bauer14} and the {\em interleaving distance} $d_{\rm
  I}$~\cite{deSilva16} share the same advantages and drawbacks as
$d_{\rm GH}$, in particular they enjoy good stability and
discriminativity properties but they lack efficient algorithms for
their computation, moreover they can be difficult to interpret.  By
contrast, the {\em bottleneck distance} $\distb$
comes with a signature for Reeb graphs, called the {\em extended
  persistence diagram}~\cite{Cohen09}, which acts as a stable
bag-of-feature descriptor. Furthermore, $\distb$ can be computed
efficiently in practice.  Its main drawback though is to be only a
pseudo-metric, so distinct graphs can have the same signature and
therefore be deemed equal in $\distb$.

Another desired property for dissimilarity measures is to be {\em intrinsic}, i.e. realized as the lengths of shortest continuous 
paths in the space of Reeb graphs~\cite{Burago01}. This is particularly useful when one actually needs to interpolate between data, 
and not just discriminate between them, which happens in applications such as 
image or 3-d shape morphing, skeletonization, and matching~\cite{Ge11, Mohamed12, Mukasa06, Tierny06}. 
At this time, it is unclear whether the metrics proposed so far for Reeb graphs are intrinsic or not.   
Using intrinsic metrics would not only open the door to the use of Reeb graphs in the aforementioned applications, 
but it would also provide a better understanding of the intrinsic structure of the space of Reeb graphs, and give a deeper meaning to the distance values.

\subparagraph*{Our contributions.}  In the first part of the paper we show that the
bottleneck distance can discriminate a Reeb graph from any other Reeb
graph in a small enough neighborhood, as efficiently as the other
metrics do, even though it is only a pseudo-metric globally.  More
precisely, we show that, given any constant $K\in(0,1/22]$, in a
  sufficiently small neighborhood of a given Reeb graph $\Reeb_f$ in
  the functional distortion distance (that is: for any Reeb graph
  $\Reeb_g$ such that $\distfd(\Reeb_f,\Reeb_g)< c(f,K)$, where
  $c(f,K)>0$ is a positive constant depending only on $f$ and $K$),
  one has:
\begin{equation}\label{eq:intro}
K\distfd(\Reeb_f,\Reeb_g)\leq\distb(\Reeb_f,\Reeb_g)\leq 2\distfd(\Reeb_f,\Reeb_g).
\end{equation}
The second inequality is already known~\cite{Bauer14}, and it asserts that the bottleneck distance between Reeb graphs is stable.  
The first inequality is new, and it asserts that the bottleneck distance is discriminative locally, 
in fact just as discriminative as the other distances mentioned above. 
Equation~(\ref{eq:intro}) can be viewed as a local equivalence between metrics
although not in the usual sense: firstly, all comparisons are anchored to a fixed Reeb graph $\Reeb_f$, 
and secondly, the constants $K$ and $2$ are absolute.

The second part of the paper advocates the study of intrinsic metrics
on the space of Reeb graphs, for the reasons mentioned above. As
a first step, we propose to study the intrinsic metrics $\hat d_{\rm
  GH}$, $\hdistfd$, $\hat d_{\rm I}$ and $\hdistb$ induced
respectively by $d_{\rm GH}$, $\distfd$, $d_{\rm I}$ and
$\distb$. While the first three are obviously globally equivalent
because their originating metrics are, our second contribution is to
show that the last one is also globally  equivalent to
the other three. 

The paper concludes with a discussion and some directions for the study of the space of Reeb graphs as an intrinsic metric space.

\subparagraph*{Related work.}

Interpolation between Reeb graphs is also the underlying idea of the
{\em edit distance} recently proposed by Di Fabio and
Landi~\cite{DiFabio16}. The problem with this distance, in its current
form at least, is that it restricts the interpolation to pairs of
graphs lying in the same homeomorphism class. By contrast, our class
of admissible paths is defined with respect to the topology induced by
the functional distortion distance, as such it allows interpolating
between distinct homeomorphism classes. 

Interpolation between Reeb
graphs is also related to the study of inverse problems in topological
data analysis.  To our knowledge, the only result in this vein shows
the differentiability of the operator sending point clouds to the
persistence diagram of their distance function~\cite{Gameiro16}.  Our
first contribution~(\ref{eq:intro}) sheds light on the operator's
local injectivity properties over the class of Reeb graphs.

\section{Background}\label{sec:Def}

Throughout the paper we work with singular homology with coefficients
in the field~$\Z_2$, which we omit in our notations for simplicity. In
the following, ``connected'' stands for ``path-connected'', and ``cc''
stands for ``connected component(s)''. Given a map $f:X\to\R$ and an interval $I\subseteq\R$, we write $X^{I}_f$ as a 
shorthand for the preimage $f^{-1}(I)$, and we omit the subscript when the map is obvious from the context.

\subsection{Morse-Type Functions}
\label{sec:Morse-type}

\begin{definition}\label{def:Morse-type}
A continuous real-valued function $f$ on a topological space $X$ 
is \emph{of Morse type} if:

(i) there is a finite set $\text{Crit}(f)=\{a_1<...<a_n\}\subset\R$, 
called the set of \emph{critical values},
such that over every open interval $(a_0=-\infty,a_1),...,(a_i,a_{i+1}),...,(a_n,a_{n+1}=+\infty)$
there is a compact and locally connected space $Y_i$  
and a homeomorphism $\mu_i:Y_i\times(a_i,a_{i+1})\rightarrow X^{(a_i,a_{i+1})}$
such that $\forall i=0,...,n, f|_{X^{(a_i,a_{i+1})}}=\pi_2\circ\mu_i^{-1}$, where
$\pi_2$ is the projection onto the second factor;

(ii) $\forall i=1,...,n-1,\mu_i$ extends to a continuous function $\bar{\mu}_i:Y_i\times[a_i,a_{i+1}]\rightarrow X^{[a_i,a_{i+1}]}$;
 similarly, $\mu_0$ extends to $\bar{\mu}_0:Y_0\times(-\infty,a_1]\rightarrow X^{(-\infty,a_1]}$
and $\mu_n$ extends to $\bar{\mu}_n:Y_n\times[a_n,+\infty)\rightarrow X^{[a_n,+\infty)}$;

(iii) Each levelset $f^{-1}(t)$ has a finitely-generated homology.
\end{definition}

Let us point out that a Morse function is also of Morse type, and that
its critical values remain critical in the definition above.  Note
that some of its regular values may be termed critical as well in this
terminology, with no effect on the analysis.

\subsection{Extended Persistence}
\label{sec:Persistence}

Let $f$ be a real-valued function on a topological space $X$.  The
family $\{X^{(-\infty, \alpha]}\}_{\alpha\in\R}$ of sublevel sets of
  $f$ defines a {\em filtration}, that is, it is nested
  w.r.t. inclusion: $X^{(-\infty, \alpha]}\subseteq X^{(-\infty,
  \beta]}$ for all $\alpha\leq\beta\in\R$. The family $\{X^{[\alpha,
      +\infty)}\}_{\alpha\in\R}$ of superlevel sets of $f$ is also
    nested but in the opposite direction: $X^{[\alpha,
        +\infty)}\supseteq X^{[\beta, +\infty)}$ for all
        $\alpha\leq\beta\in\R$. We can turn it into a filtration by
        reversing the order on the real line. Specifically, let
        $\Rop=\{\tilde{x}\ |\ x\in\R\}$, ordered by
        $\tilde{x}\leq\tilde{y}\Leftrightarrow x\geq y$. We index the
        family of superlevel sets by $\Rop$, so now we have a
        filtration: $\{X^{[\tilde\alpha,
            +\infty)}\}_{\tilde\alpha\in\Rop}$, with
          $X^{[\tilde\alpha, +\infty)}\subseteq X^{[\tilde\beta,
                +\infty)}$ for all
              $\tilde\alpha\leq\tilde\beta\in\Rop$.

Extended persistence connects the two filtrations at infinity as
follows. First, replace each superlevel set $X^{[\tilde\alpha, +\infty)}$ by
  the pair of spaces $(X, X^{[\tilde\alpha, +\infty)})$ in the second
    filtration. This maintains the filtration property since we have
    $(X, X^{[\tilde\alpha, +\infty)})\subseteq (X, X^{[\tilde\beta,
          +\infty)})$ for all
        $\tilde\alpha\leq\tilde\beta\in\Rop$. Then, let
        $\Rext=\R\cup\{+\infty\}\cup\Rop$, where the order is
        completed by $\alpha<+\infty<\tilde{\beta}$ for all
        $\alpha\in\R$ and $\tilde\beta\in\Rop$. This poset is
        isomorphic to $(\R, \leq)$. Finally, define the {\em extended
          filtration} of $f$ over $\Rext$ by:
\[
F_\alpha = X^{(-\infty, \alpha]}\ \mbox{for $\alpha\in\R$},\ F_{+\infty} = X \equiv (X,\emptyset)\text{ and }
F_{\tilde\alpha} = (X, X^{[\tilde\alpha, +\infty)})\ \mbox{for $\tilde\alpha\in\Rop$},
\]
where we have identified the space $X$ with the pair of spaces $(X,
\emptyset)$ at infinity. 
The subfamily $\{F_\alpha\}_{\alpha\in\R}$ is
  the \textit{ordinary} part of the filtration, while $\{F_{\tilde\alpha}\}_{\tilde\alpha\in\Rop}$ is the
  \textit{relative} part. 

Applying the homology functor $H_*$ to this filtration gives the so-called \textit{extended persistence module} $\mathbb{V}$ of $f$,
which is a sequence of vector spaces connected by linear maps induced by the
inclusions in the extended filtration.
For functions of Morse type, the extended persistence module can be
decomposed as a finite direct sum of half-open {\em interval
  modules}---see e.g.~\cite{Chazal16}:
$\mathbb{V}\simeq\bigoplus_{k=1}^n \mathbb{I}[b_k, d_k)$,
where each summand $\mathbb{I}[b_k, d_k)$ is made of copies of the
  field of coefficients at every index $\alpha\in [b_k, d_k)$, and of
    copies of the zero space elsewhere, the maps between copies of the
    field being identities.  Each summand represents the lifespan of a
    {\em homological feature} (cc, hole, void, etc.) within the
    filtration. More precisely, the {\em birth time} $b_k$ and {\em
      death time} $d_k$ of the feature are given by the endpoints of
    the interval.  Then, a convenient way to represent the structure
    of the module is to plot each interval in the decomposition as a
    point in the extended plane, whose coordinates are given by the
    endpoints. Such a plot is called the \textit{extended persistence
      diagram} of $f$, denoted $\Dg(f)$.  The distinction
    between ordinary and relative parts of the filtration allows us to
    classify the points in $\Dg(f)$ as follows:
\begin{itemize}
\item $p=(x,y)$ is called an {\em ordinary} point if $x,y\in\R$; 
\item $p=(x,y)$ is called a {\em relative} point if $x,y\in\Rop$;
\item $p=(x,y)$ is called an {\em extended} point if $x\in\R,y\in\Rop$; 
\end{itemize}
Note that ordinary points lie strictly above the diagonal
$\Delta=\{(x,x)\ |\ x\in\R\}$ and relative points lie strictly below $\Delta$,
while extended points can be located anywhere, including on $\Delta$ (e.g. when a cc 
lies inside a single critical level, see Section~\ref{sec:ReebGraphDef}). 
It is common to
partition $\Dg(f)$ according to this classification:
$\Dg(f)=\Ord(f)\sqcup\Rel(f)\sqcup\Ext^+(f)\sqcup\Ext^-(f)$, where by
convention $\Ext^+(f)$ includes the extended points located on the
diagonal~$\Delta$.

\subparagraph*{Stability.}
An important property of extended persistence diagrams is to be stable
in the so-called {\em bottleneck distance}~$d^\infty_{\emph{b}}$.  
Given two persistence diagrams $D,D'$, a \emph{partial matching} between $D$ and $D'$
is a subset $\Gamma$ of $D\times D'$ where for every
$p\in D$ there is at most one $p'\in D'$ such that $(p,p')\in\Gamma$,
and conversely, for every $p'\in D'$ there is at most one $p\in D$
such that $(p,p')\in\Gamma$.
Furthermore, $\Gamma$ must match points of the same type (ordinary,
relative, extended) and of the same homological dimension only.  
The \emph{cost} of $\Gamma$ is:
$\text{cost}(\Gamma)=\max\{\displaystyle\max_{p\in
  D}\ \delta_D(p),\ \max_{p'\in D'}\ \delta_{D'}(p')\}$, where
$\delta_D(p)=\|p-p'\|_\infty$ if $p$ is matched to some $p'\in D'$
 and $\delta_D(p)=d_\infty(p,\Delta)$
if $p$ is unmatched---same for $\delta_{D'}(p')$.
\begin{definition}
\label{def:bottleneck}
The \emph{bottleneck distance} between two persistence diagrams $D$ and $D'$ is
$\distb(D,D')=\inf_{\Gamma}\text{cost}(\Gamma)$,
where $\Gamma$ ranges over all partial matchings between $D$ and $D'$.
\end{definition}

\begin{thm}[Stability~\cite{Cohen09}]
\label{th:Stab}
For any  Morse-type functions $f,g:X\rightarrow\R$,
\begin{equation}
\distb(\Dg(f),\Dg(g))\leq \|f-g\|_\infty.
\end{equation}
\end{thm}

\subsection{Reeb Graphs}
\label{sec:ReebGraphDef}

\begin{definition}
Given a topological space $X$ and a continuous function 
$f:X\rightarrow\R$, we define the equivalence relation $\sim_f$ between points of $X$ by
$x\sim_f y$
if and only if
$f(x)=f(y)$ 
and $x,y$ belong to the same cc of $f^{-1}(f(x))=f^{-1}(f(y))$.
The \emph{Reeb graph} $\Reeb_f(X)$ is the quotient space $X/\sim_f$.
As $f$ is constant on equivalence classes, there is a well-defined induced
 map $\tilde f:\Reeb_f(X)\rightarrow\R$.
\end{definition}

\subparagraph*{Connection to extended persistence.}
If $f$ is a function of Morse type, then
the pair $(X,f)$ is an {\em $\R$-constructible space} in the sense
of~\cite{deSilva16}.  This ensures that the Reeb graph is a
multigraph, whose nodes are in one-to-one correspondence with the cc
of the critical level sets of $f$.  
In that case, there is a nice interpretation of $\Dg(\tilde f)$ in terms of the
structure of $\Reeb_f(X)$. We refer the reader to~\cite{Bauer14, Cohen09} and
the references therein for a full description as well as formal
definitions and statements.  
Orienting the Reeb graph vertically so ${\tilde f}$ is the height function,
we can see each cc of the graph as a trunk with
multiple branches (some oriented upwards, others oriented downwards)
and holes.  Then, one has the following correspondences, where the
{\em vertical span} of a feature is the span of its image by~$\tilde
f$:
\begin{itemize}
\item The vertical spans of the trunks
are given by the points in $\Ext_0^+(\tilde f)$;
\item The vertical spans of the downward branches are given by the
  points in $\Ord_0(\tilde f)$;
\item The vertical spans of the upward branches are given by the
  points in $\Rel_1(\tilde f)$;
\item The vertical spans of the holes are given by the points in
  $\Ext_1^-(\tilde f)$.
\end{itemize}
The rest of the diagram of~$\tilde f$ is empty.  These correspondences
provide a dictionary to read off the structure of the Reeb graph from
the persistence diagram of the quotient map~$\tilde f$. Note that it is a
bag-of-features type of descriptor, taking an inventory of all the
features together with their vertical spans, but leaving aside the actual
layout of the features. As a consequence, it is an incomplete
descriptor: two Reeb graphs with the same persistence diagram may not
be isomorphic. See the two Reeb graphs in Figure~\ref{fig:cetrue} for instance.

\subparagraph*{Notation.}  Throughout the paper, we consider Reeb graphs coming from Morse-type functions, equipped with their induced maps. We denote by $\RS$ the space of such graphs. 
In the following, we have $\Reeb_f,\Reeb_g \in\RS$, with induced maps
$f:\Reeb_f\rightarrow\R$ with critical values $\{a_1,...,a_n\}$, and
$g:\Reeb_g\rightarrow\R$ with critical values $\{b_1,...,b_m\}$.  Note
that we write $f,g$ instead of $\tilde f,\tilde g$ for convenience.
We also assume without loss of generality (w.l.o.g.) that $\Reeb_f$
and $\Reeb_g$ are connected. 
If they are not connected, then our analysis
can be applied component-wise.  

\subsection{Distances for Reeb graphs}

\begin{definition}
The {\em bottleneck distance} between $\Reeb_f$ and $\Reeb_g$ is:
\begin{equation}\label{eq:dbott}
\distb(\Reeb_f,\Reeb_g):=\distb(\Dg(f),\Dg(g)).
\end{equation}
\end{definition}

\begin{definition}
The {\em functional distortion distance} between $\Reeb_f$ and $\Reeb_g$ is:
\begin{equation}\label{eq:dfd}
\distfd(\Reeb_f,\Reeb_g):=\underset{\phi,\psi}{\rm inf}\ {\rm max}\left\{\frac{1}{2}D(\phi,\psi),\ \|f-g\circ\phi\|_\infty,\ \|f\circ\psi-g\|_\infty\right\},
\end{equation}
where: 
\begin{itemize}
\item $\phi:\Reeb_f\rightarrow\Reeb_g$ and $\psi:\Reeb_g\rightarrow\Reeb_f$ are continuous maps, 
\item $D(\phi,\psi)=\ {\rm sup}\ \left\{|d_f(x,x')-d_g(y,y')|\ {\rm such\ that}\ (x,y),(x',y')\in C(\phi,\psi)\right\},$ where:
\begin{itemize}
\item $C(\phi,\psi)=\{(x,\phi(x))\ |\ x\in\Reeb_f\}\cup\{(\psi(y),y)\ |\ y\in\Reeb_g\}$, 
\item $d_f(x,x')=\underset{\pi:x\rightarrow x'}{\rm min}\ \left\{\underset{t\in[0,1]}{\max}\ f\circ\pi(t)-\underset{t\in[0,1]}{\min}\ f\circ\pi(t)\right\}$, 
where $\pi:[0,1]\rightarrow\Reeb_f$ is a continuous path from $x$ to $x'$ in $\Reeb_f$ ($\pi(0)=x$ and $\pi(1)=x'$),
\item $d_g(y,y')=\underset{\pi:y\rightarrow y'}{\rm min}\ \left\{\underset{t\in[0,1]}{\max}\ g\circ\pi(t)-\underset{t\in[0,1]}{\min}\ g\circ\pi(t)\right\}$, 
where $\pi:[0,1]\rightarrow\Reeb_g$ is a continuous path from $y$ to $y'$ in $\Reeb_g$ ($\pi(0)=y$ and $\pi(1)=y'$).
\end{itemize}
\end{itemize}
\end{definition}

Bauer et al.~\cite{Bauer14} related these distances as follows:
\begin{thm}\label{th:lowerboundw}
The following inequality holds: 
$\distb(\Reeb_f,\Reeb_g)
\leq 3\,\distfd(\Reeb_f,\Reeb_g).$
\end{thm}
This result can be improved using the end of Section~3.4 of~\cite{Bjerkevik16}, then noting that 
level set diagrams and extended diagrams are essentially the same~\cite{Carlsson09b}, and finally Lemma~9 of~\cite{Bauer15}:
\begin{thm}\label{th:lowerbound}
The following inequality holds: 
$\distb(\Reeb_f,\Reeb_g)
\leq 2\,\distfd(\Reeb_f,\Reeb_g).$
\end{thm}
We emphasize that, even though Theorem~\ref{th:lowerbound} allows us to improve on the constants of our main result---see Theorem~\ref{th:locisom},
the reduction from $3\,\distfd(\Reeb_f,\Reeb_g)$ in Theorem~\ref{th:lowerboundw} to $2\,\distfd(\Reeb_f,\Reeb_g)$  
in Theorem~\ref{th:lowerbound}
is not fundamental for our analysis and proofs.

Since the bottleneck distance is only a pseudo-metric---see Figure~\ref{fig:cetrue}, the inequality 
given by Theorem~\ref{th:lowerbound}
cannot be turned into an equivalence result. However, for any pair of Reeb graphs $\Reeb_f,\Reeb_g$ 
that have the same extended persistence diagram $\Dg(f)=\Dg(g)$, 
and that are at positive functional distortion distance from each other, 
every continuous path in $\distfd$ from $\Reeb_f$ to $\Reeb_g$ will perturb the points of $\Dg(f)$ 
and eventually drive them back to their initial position,
suggesting first that $\distb$ is locally equivalent to $\distfd$---see Theorem~\ref{th:locisom} in Section~\ref{sec:loc}, but also that, 
even though $\distb(\Reeb_f,\Reeb_g)=0$, 
the intrinsic metric $\hdistb(\Reeb_f,\Reeb_g)$ induced by $\distb$ is positive---see Theorem~\ref{th:strongeq} in Section~\ref{sec:induced}.  

\begin{figure}[h]\centering
\includegraphics[width=12cm]{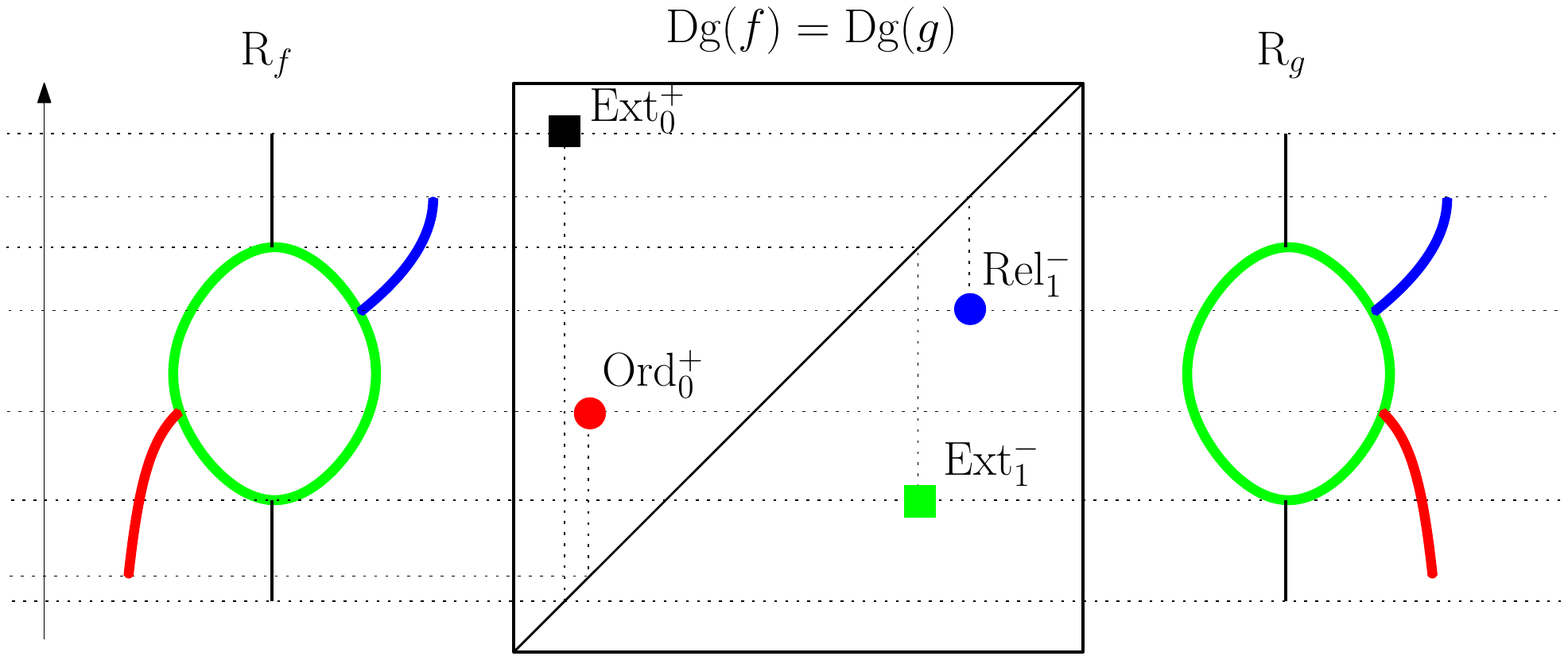}
\caption{\label{fig:cetrue} Example of two different Reeb graphs $\Reeb_f$ and $\Reeb_g$ that have the same extended persistence diagram $\Dg(f)=\Dg(g)$.
These graphs are at bottleneck distance 0 from each other, while their functional distortion distance is positive.}
\end{figure}

\section{Local Equivalence}\label{sec:loc}

Let $a_f=\min_{1\leq i\leq n}a_{i+1}-a_i >0$
and $a_g=\min_{1\leq j\leq m}b_{j+1}-b_j >0$.
In this section, 
we show the following local equivalence theorem:

\begin{thm}\label{th:locisom}
Let $K\in(0,1/22]$. If $\distfd(\Reeb_f,\Reeb_g)\leq\max\{a_f,a_g\}/(8(1+22K))$, then:
$$K\distfd(\Reeb_f,\Reeb_g)\leq\distb(\Reeb_f,\Reeb_g)\leq 2\distfd(\Reeb_f,\Reeb_g).$$ 
\end{thm}

Note that the notion of locality used here is slightly different from
the usual one.  On the one hand, the equivalence does not hold for any
arbitrary pair of Reeb graphs inside a neighborhood of some fixed Reeb
graph, but rather for any pair involving the fixed graph. 
On the other hand, the constants in the equivalence are independent of the pair of Reeb graphs considered. 
The upper bound on $\distb(\Reeb_f,\Reeb_g)$ is given by Theorem~\ref{th:lowerbound}
and always holds.
The aim of this section is to prove the lower bound.

{\bf Convention:}
We assume w.l.o.g. that $\max\{a_f,a_g\}=a_f$, 
and we let $\e=\distfd(\Reeb_f,\Reeb_g)$.

\subsection{Proof of Theorem~\ref{th:locisom}}
\label{sec:geomgraphs}

Let $K\in(0,1/22]$.
The proof proceeds by contradiction. Assuming $\distb(\Reeb_f, \Reeb_g) < K\e$, 
where $\e=\distfd(\Reeb_f, \Reeb_g)< a_f/(8(1+22K))$, 
we  progressively transform $\Reeb_g$ into some other Reeb graph $\Reeb_{g'}$ 
(Definition~\ref{def:simpl}) that satisfies both $\distfd(\Reeb_g, \Reeb_{g'})< \e$ (Proposition~\ref{prop:fdb}) and   
$\distfd(\Reeb_f, \Reeb_{g'})=0$ (Proposition~\ref{prop:iso2}). The contradiction follows from the triangle inequality. 

\subsubsection{Graph Transformation}\label{sec:smoothope}

The graph transformation is defined as the composition of the {\em
  simplification operator} from~\cite{Bauer14} and the {\em Merge
  operator}\footnote{Strictly speaking, the output of our Merge
   is the Reeb graph of the output of 
the Merge from~\cite{Carriere16}.}
from~\cite{Carriere16}.  We refer the reader to these articles for the
precise definitions. Below we merely recall their main properties.
Given a set $S\subseteq X$ and a scalar $\alpha >0$, we recall that
$S^\alpha=\{x\in X\ |\ d(x,S)\leq\alpha\}$ denotes the $\alpha$-offset
of $S$.

\begin{lemma}[Theorem~7.3 and following remark in~\cite{Bauer14}] \label{lem:stabsmooth}
Given $\alpha>0$, the simplification operator $S_\alpha:\RS \to \RS$ takes any Reeb graph 
$\Reeb_h$ to $\Reeb_{h'}=S_\alpha(\Reeb_h)$ such that $\Dg(h')\cap \Delta^{\alpha/2}=\emptyset$  
and $\distb(\Reeb_h,\Reeb_{h'})\leq 2\,\distfd(\Reeb_h,\Reeb_{h'}) \leq 4\alpha$.
\end{lemma} 

\begin{lemma}[Theorem~2.5 and Lemma~4.3 in~\cite{Carriere16}] \label{lem:merge}
Given $a\leq b$, the merge operator $\Merge_{a, b}: \RS\to\RS$ takes any Reeb graph 
$\Reeb_h$ to $\Reeb_{h'}=\Merge_{a,b}(\Reeb_h)$ such that 
$\Dg(h')$ is obtained from $\Dg(h)$ through the following snapping principle (see Figure~\ref{fig:merge} for an illustration):
\[
(x,y)\in\Dg(h) \mapsto (x',y')\in\Dg(h')\text{ where } x' =
\left\{ \begin{array}{ll} x\text{ if }x\notin[a,b] \\[0.5ex]
\frac{a+b}{2}\text{ otherwise } \end{array} \right.
\text{ and similarly for }y'.
\]
\end{lemma} 

\begin{figure}[htb]\centering
\includegraphics[height=5cm]{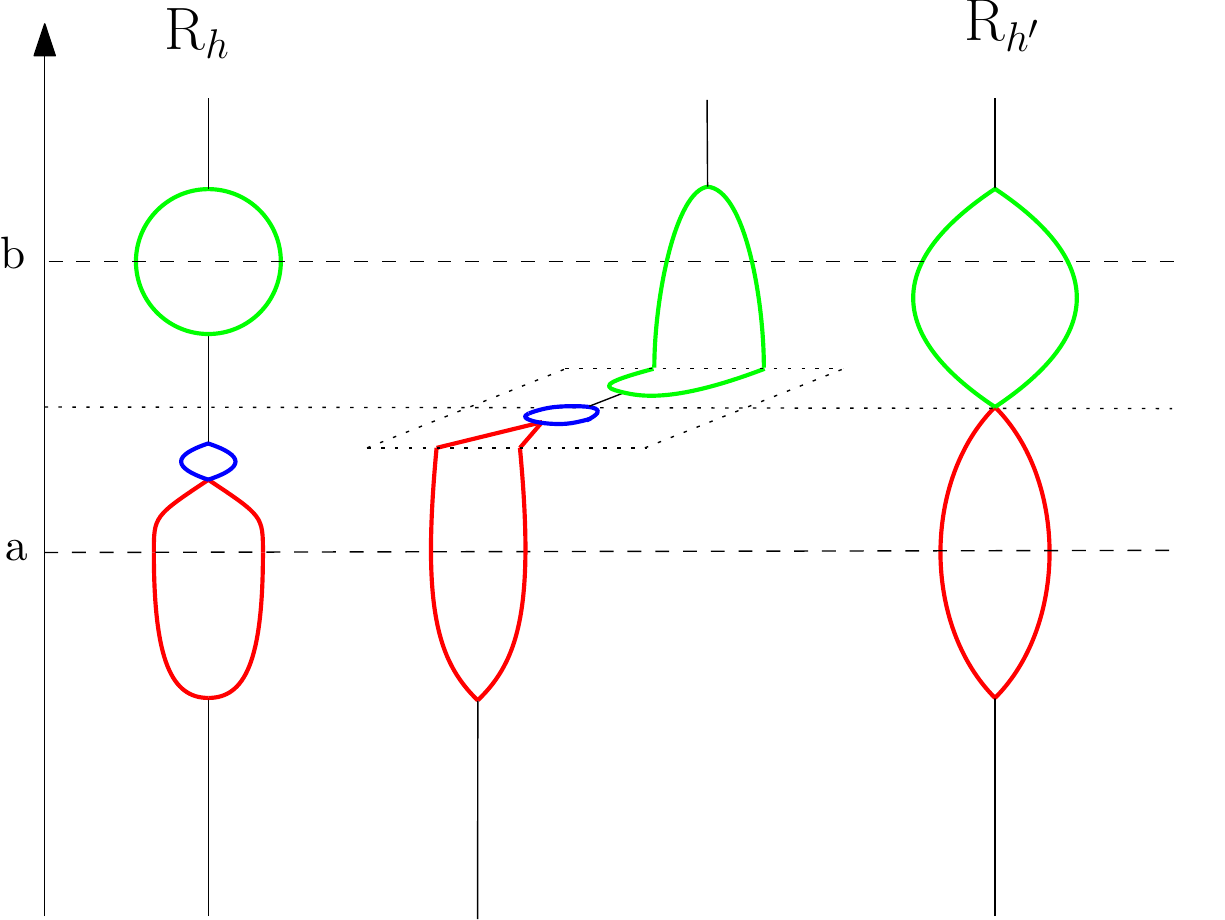}
\hspace{1cm}
\includegraphics[height=5cm]{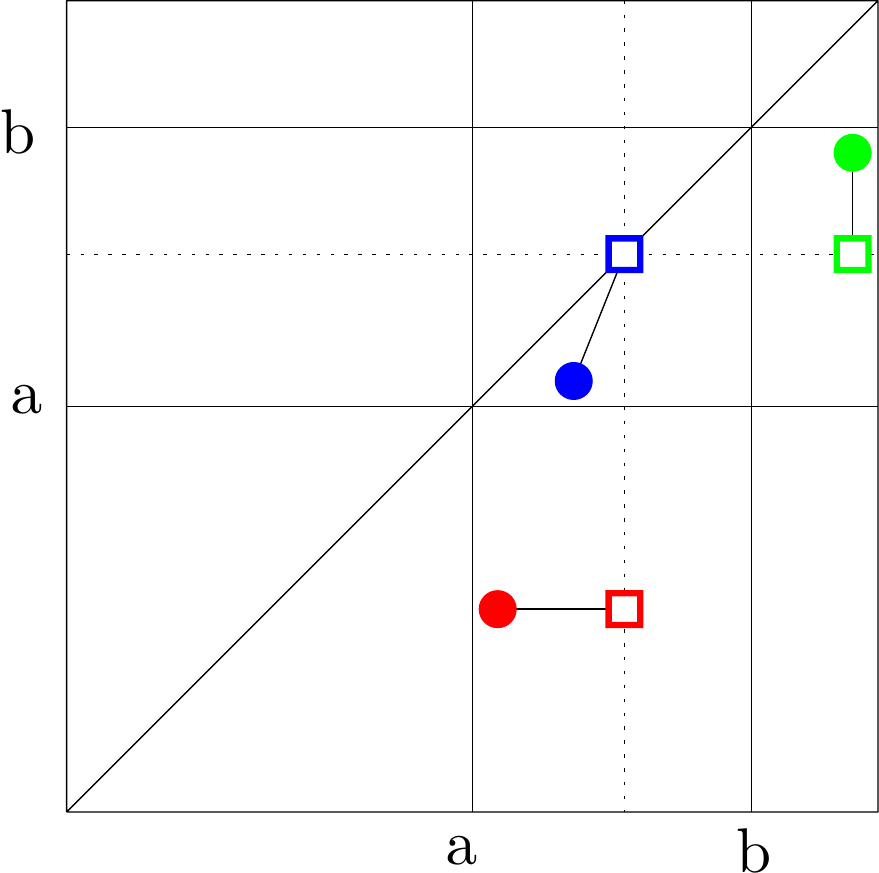}
\caption{\label{fig:merge} 
Left: effect of $\Merge_{a,b}$ on a Reeb graph~$\Reeb_h$.
Right: Effect
on its persistence diagram.}
\end{figure}

\begin{definition}\label{def:simpl}
Let $\Reeb_f$ be a fixed Reeb graph with critical values $\{a_1,
\cdots, a_n\}$. Given $\alpha > 0$, the full transformation
$F_\alpha:\RS\to\RS$ is defined as
$
F_{\alpha}=
\Merge_{9\alpha}
\circ S_{2\alpha},
$
where $\Merge_{9\alpha}=\Merge_{a_n-9\alpha,\:a_n+9\alpha}\circ ... \circ \Merge_{a_1-9\alpha,\:a_1+9\alpha}$.
See Figure~\ref{fig:recap3} for an illustration.
\end{definition}

\begin{figure}[htb]
\centering
\includegraphics[width=11cm]{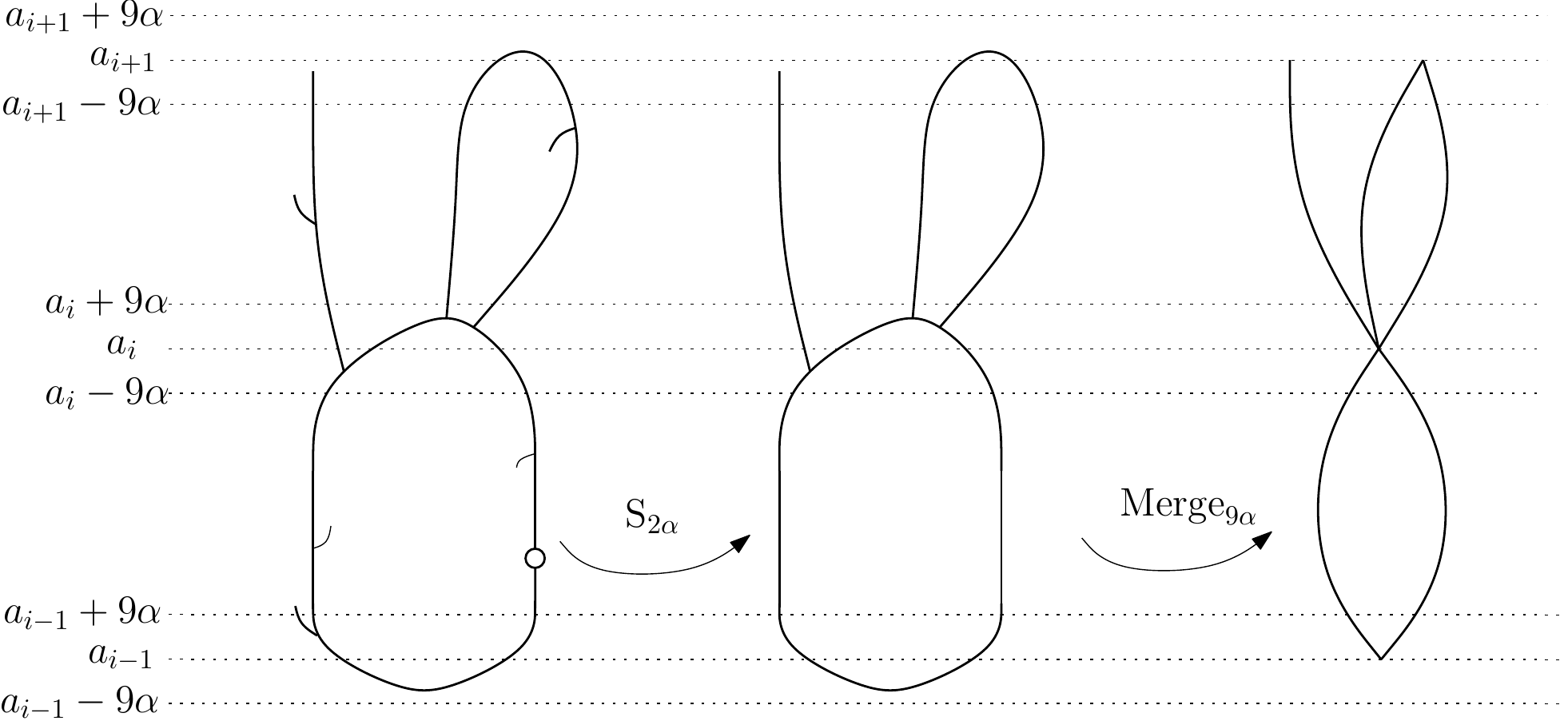}
\caption{Illustration of $F_{\alpha}$.}
\label{fig:recap3}
\end{figure}

\subsubsection{Properties of the transformed graph}

Let $\Reeb_f, \Reeb_g\in\RS$ such that $\distb(\Reeb_f, \Reeb_g) <
K\e$ where $\e=\distfd(\Reeb_f, \Reeb_g)< a_f/(8(1+22K))$.  Letting
$\Reeb_{g'} = F_{K\e}(\Reeb_g)$, we want to show both that
$\distfd(\Reeb_g,\Reeb_{g'})<22K\e<\e$ and $\distfd(\Reeb_f,\Reeb_{g'})=0$,
which will lead to a contradiction as mentioned previously.  

Let $B_\infty(\cdot,\cdot)$ denote balls in the $\ell_\infty$-norm.
\begin{lemma}\label{lem:crit}
Let $\Reeb_h=S_{2K\e}(\Reeb_g)$.
Under the above assumptions, one has 
\begin{equation}\label{eq:crit}
\Dg(h)\subseteq{\bigcup}_{\tau\in\Dg(f)} B_\infty(\tau,9K\e).
\end{equation}
\end{lemma}

\begin{proof}
Since $\distb(\Reeb_f,\Reeb_g)<K\e$, we have
$\Dg(g)\subseteq \bigcup_{\tau\in\Dg(f)} B_\infty(\tau,K\e) \cup \Delta^{K\e}$.
Since $\Reeb_h=S_{2K\e}(\Reeb_g)$,
it follows from Lemma~\ref{lem:stabsmooth} that $\distb(\Dg(h),\Dg(g))\leq 8K\e$.
Moreover, since every persistence pair in $\Dg(g)\cap\Delta^{K\e}$ is removed by $S_{2K\e}$,
it results that: 

$\Dg(h)\subseteq{\bigcup}_{\tau\in\Dg(g)\setminus \Delta^{K\e}} B_\infty(\tau,8K\e)\subseteq{\bigcup}_{\tau\in\Dg(f)} B_\infty(\tau,9K\e).$
\end{proof}

Now we bound $\distfd(\Reeb_{g'},\Reeb_g)$. 
Recall that, given an arbitrary Reeb graph $\Reeb_h$, with critical values $\Crit(h)=\{c_1,...,c_p\}$,
if $C$ is a cc of $h^{-1}(I)$, where $I$ is an open interval such that 
$\exists c_i,c_{i+1}\ {\rm s.t.}\ I\subseteq(c_i,c_{i+1})$, 
then $C$ is a {\em topological arc}, i.e. homeomorphic to an open interval.	

\begin{proposition}\label{prop:fdb}
Under the same assumptions as above,  
one has $\distfd(\Reeb_g,\Reeb_{g'})<22K\epsilon$.
\end{proposition}

\begin{figure}[htb]\centering
\includegraphics[width=11.6cm]{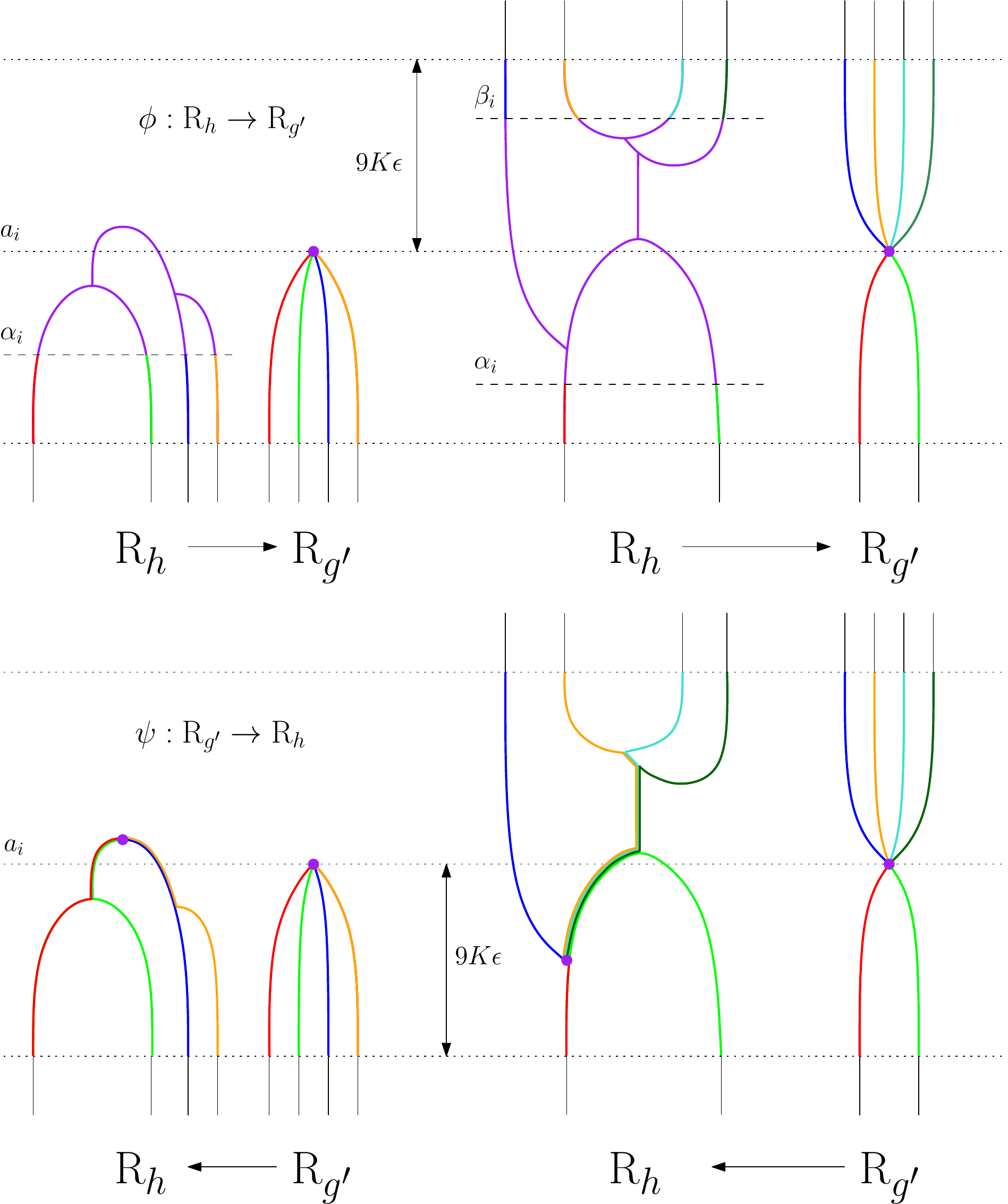}
\caption{\label{fig:fdmerge}
The effects of $\phi$ and $\psi$ around a specific critical value $a_i$ of $f$. Segments are matched according
to their colors (up to reparameterization).
}
\end{figure}

\begin{proof}
Let $\Reeb_h=S_{2K\e}(\Reeb_g)$.
We have
$
\distfd(\Reeb_{g'},\Reeb_g)\leq \distfd(\Reeb_{g'},\Reeb_h)+\distfd(\Reeb_h,\Reeb_g)
$
by the triangle inequality.
It suffices therefore to bound both
$\distfd(\Reeb_{g'},\Reeb_h)$ and $\distfd(\Reeb_h,\Reeb_{g})$.
By Lemma~\ref{lem:stabsmooth}, we have $\distfd(\Reeb_h,\Reeb_g) < 4K\e$.
Now, recall from~(\ref{eq:crit}) that the points of the extended
persistence diagram of $\Reeb_h$ are included in
${\bigcup}_{\tau\in\Dg(f)} B_\infty(\tau,9K\e)$.  Moreover, since
$\Reeb_{g'}=\Merge_{9K\e}(\Reeb_h)$, $\Reeb_{g'}$ and $\Reeb_h$ are
composed of the same number of arcs in each $[a_i+9K\e,a_{i+1}-9K\e]$.
Hence, we can define explicit continuous maps
$\phi:\Reeb_h\rightarrow\Reeb_{g'}$ and
$\psi:\Reeb_{g'}\rightarrow\Reeb_h$ as depicted in
Figure~\ref{fig:fdmerge}.  More precisely, since $\Reeb_h$ and
$\Reeb_{g'}$ are composed of the same number of arcs in each
$[a_i+9K\e,a_{i+1}-9K\e]$, we only need to specify $\phi$ and $\psi$
inside each interval $(a_i-9K\e,a_i+9K\e)$ and then ensure that the 
piecewise-defined maps are assembled consistently.  Since the critical values
of $\Reeb_h$ are within distance less that $9K\e$ of the critical
values of $f$, there exist two levels $a_i-9K\e < \alpha_i\leq\beta_i
< a_i+9K\e$ such that $\Reeb_h$ is only composed of arcs in
$(a_i-9K\e,\alpha_i]$ and $[\beta_i,a_i+9K\e)$ for each $i$ (dashed
    lines in Figure~\ref{fig:fdmerge}).  For any cc~$C$ of
    $h^{-1}((a_i-9K\e,a_i+9K\e))$, the map $\phi$ sends all points of
    $C\cap h^{-1}([\alpha_i,\beta_i])$ to the corresponding critical
    point $y_C$ created by the Merge in $\Reeb_{g'}$, and it extends the
    arcs of $C\cap h^{-1}((a_i-9K\e,\alpha_i])$ (resp. $C\cap
    h^{-1}([\beta_i,a_i+9K\e))$) into arcs of $(g')^{-1}([a_i-9K\e,a_i])$
    (resp. $(g')^{-1}([a_i,a_i+9K\e])$). 
    In return, the map $\psi$ sends the critical
    point~$y_C$ to an arbitrary point of $C$.  Then, since the Merge
    operation preserves connected components, for each arc $A'$ of
    $(g')^{-1}((a_i-9K\e,a_i+9K\e))$ connected to~$y_C$, there is at
    least one corresponding path $A$ in $\Reeb_h$ whose endpoint in
    $h^{-1}(a_i-9K\e)$ or $h^{-1}(a_i+9K\e)$ matches with the one of $A'$
    (see the colors in Figure~\ref{fig:fdmerge}).  Hence $\psi$ sends
    $A'$ to $A$ and the piecewise-defined maps
    are assembled consistently.

\noindent Let us bound the
three terms in the $\max\{\cdots\}$ in~(\ref{eq:dfd})
with this choice of maps $\phi, \psi$: 
\begin{itemize}
\item We first  bound $\|g'-h\circ\psi\|_\infty$. Let $x\in\Reeb_{g'}$. Either $g'(x)\in \bigcup_{i\in\{1,...,n-1\}} [a_i+9K\e,a_{i+1}-9K\e]$,  
and in this case we have $g'(x)=h(\psi(x))$ by definition of~$\psi$;
or, there is $i_0\in\{1,...,n\}$ such that $g'(x)\in (a_{i_0}-9K\e,a_{i_0}+9K\e)$ and then $h(\psi(x))\in (a_{i_0}-9K\e,a_{i_0}+9K\e)$.
In both cases $|g'(x)-h\circ\psi(x)| < 18K\e$. Hence, $\|g'-h\circ\psi\|_\infty < 18K\e$.
\item Since the previous proof is symmetric in $h$ and $g'$, one also has $\|h-g'\circ\phi\|_\infty < 18K\e$.
\item We now bound $D(\phi,\psi)$. Let $(x,\phi(x)),(\psi(y),y)\in C(\phi,\psi)$ (the cases $(x,\phi(x)),(x',\phi(x'))$
and $(\psi(y),y),(\psi(y'),y')$ are similar). Let $\pi_{g'}:[0,1]\rightarrow\Reeb_{g'}$ be a continuous path from $\phi(x)$ to $y$  
which achieves $d_{g'}(\phi(x),y)$.
\begin{itemize}
\item Assume $h(x)\in \bigcup_{i\in\{1,...,n-1\}} [a_i+9K\e,a_{i+1}-9K\e]$. Then one has $\psi\circ\phi(x)=x$.
Hence, $\pi_h:=\psi\circ\pi_{g'}$ is a valid path from $x$ to $\psi(y)$. Moreover, since $\|g'-h\circ\psi\|_\infty < 18K\e$,
it follows that
\begin{equation}\label{eqs:noname}
\begin{array}{l}\max\ \im(h\circ \pi_{h}) < \max\ \im(g'\circ \pi_{g'}) +18K\e,\\[0.5ex]
\min\ \im(h\circ\pi_{h}) > \min\ \im(g'\circ\pi_{g'}) -18K\e.\end{array}
\end{equation}
Hence, one has 
$$\begin{array}{l}d_h(x,\psi(y))\leq \max\ \im(h\circ\pi_{h})-\min\ \im(h\circ\pi_{h}) < d_{g'}(\phi(x),y)+36K\e,\\
-d_h(x,\psi(y))\geq \min\ \im(h\circ\pi_{h})-\max\ \im(h\circ\pi_{h}) > -d_{g'}(\phi(x),y)-36K\e.\end{array}$$
This shows that $|d_h(x,\psi(y))-d_{g'}(\phi(x),y)| < 36K\e$.
\item Assume that there is $i_0\in\{1,...,n\}$ such that $h(x)\in(a_{i_0}-9K\e,a_{i_0}+9K\e)$.
Then, by definition of $\phi,\psi$, we have  $g'(\phi(x))\in(a_{i_0}-9K\e,a_{i_0}+9K\e)$, and, since $\phi$ and $\psi$ preserve connected components,
there is a path
$\pi'_{h}:[0,1]\rightarrow\Reeb_h$ from $x$ to $\psi\circ\phi(x)$
within the interval $(a_{i_0}-9K\e,a_{i_0}+9K\e)$, which itself is
included in the interior of the offset ${\rm
  im}(g'\circ\pi_{g'})^{18K\e}$. Let now
$\pi_h$ be the concatenation of $\pi'_h$ with $\psi\circ\pi_{g'}$, which goes from $x$ to
$\psi(y)$.
Since $\|g'-h\circ\psi\|<18K\e$, it follows that ${\rm
  im}(h\circ\psi\circ\pi_{g'})\subseteq{\rm int}\ {\rm
  im}(g'\circ\pi_{g'})^{18K\e}$, and since ${\rm im}(h\circ\pi_h)={\rm
  im}(h\circ\pi'_h)\cup{\rm im}(h\circ\psi\circ\pi_{g'})$ by
concatenation, one finally has ${\rm im}(h\circ\pi_h)\subseteq {\rm
  int}\ {\rm im}(g'\circ\pi_{g'})^{18K\e}.$ Hence, the 
inequalities of~(\ref{eqs:noname}) hold, implying that $|d_h(x,\psi(y))-d_{g'}(\phi(x),y)| < 36K\e$.
\end{itemize} 
Since these inequalities hold for any couples $(x,\phi(x))$ and $(\psi(y),y)$, we deduce that $D(\phi,\psi) \leq 36K\e$.
\end{itemize}
Thus,
$\distfd(\Reeb_{h},\Reeb_g)<4K\e$  and $\distfd(\Reeb_h,\Reeb_{g'})\leq 18K\e$, so
$\distfd(\Reeb_{g'},\Reeb_g) < 22K\e$ as desired.
\end{proof}

Now we show that $\Reeb_{g'}$ is isomorphic to $\Reeb_f$ (i.e. it lies at functional distortion distance~$0$).  

\begin{proposition}\label{prop:iso2}
Under the same assumptions as above,  
one has $\distfd(\Reeb_f,\Reeb_{g'})=0$.
\end{proposition}

\begin{proof}
First, recall from~(\ref{eq:crit}) that the points of the extended
persistence diagram of $\Reeb_h$ are included in
${\bigcup}_{\tau\in\Dg(f)} B_\infty(\tau,9K\e)$.
Since $\Reeb_{g'}=\Merge_{9K\e}(\Reeb_h)$, 
it follows from Lemma~\ref{lem:merge} that $\Crit(g')\subseteq\Crit(f)$.
Hence, both $\Reeb_{g'}$ and $\Reeb_f$ are composed of arcs in each $(a_i,a_{i+1})$.

Now, we show that, for each $i$, the number of arcs of $(g')^{-1}((a_i,a_{i+1}))$
and $f^{-1}((a_i,a_{i+1}))$ are the same.
By the triangle inequality and Proposition~\ref{prop:fdb}, we have: 
\begin{equation}\label{eq:tri}
\distfd(\Reeb_f,\Reeb_{g'}) \leq  \distfd(\Reeb_f,\Reeb_g)  + \distfd(\Reeb_g,\Reeb_{g'}) < (1+22K)\e.
\end{equation}
Let $\phi:\Reeb_f\rightarrow\Reeb_{g'}$ and $\psi:\Reeb_{g'}\rightarrow\Reeb_f$
be optimal continuous maps that achieve $\distfd(\Reeb_f,\Reeb_{g'})$.
Let $i\in\{1,...,n-1\}$.
Assume that there are more arcs of $f^{-1}((a_i,a_{i+1}))$ 
than arcs of $(g')^{-1}((a_i,a_{i+1}))$. For every arc $A$ of
$f^{-1}((a_i,a_{i+1}))$, let $x_A\in A$ such that $f(x_A)=\bar{a}=\frac{1}{2}(a_i+a_{i+1})$.
First, note that $\phi(x_A)$ must belong to an arc of 
$(g')^{-1}((a_i,a_{i+1}))$. Indeed, since $\|f-g'\circ\phi\|_\infty < (1+22K\e)$,
one has $g'(\phi(x_A))\in(\bar{a}-(1+22K)\e,\bar{a}+(1+22K)\e)\subseteq(a_i,a_{i+1})$.
Then, according to the pigeonhole principle,
there exist $x_A,x_{A'}$ such that $\phi(x_A)$ and $\phi(x_{A'})$
belong to the same arc of $(g')^{-1}((a_i,a_{i+1}))$.
\begin{itemize}
\item Since $x_A$ and $x_{A'}$ do not belong to the same arc, we have
$d_f(x_{A},x_{A'}) > a_f/2.$

\item Now, since $\|f-g'\circ\phi\|_\infty<(1+22K)\e$ and $\phi(x_{A}),\phi(x_{A'})$ belong to the same arc of $(g')^{-1}((a_i,a_{i+1}))$,
we also have 
$d_{g'}(\phi(x_{A}),\phi(x_{A'})) < 2(1+22K)\e$ (see Figure~\ref{fig:smallBranches}).
\end{itemize}
\begin{figure}[htb]
\centering
\includegraphics[width=10cm]{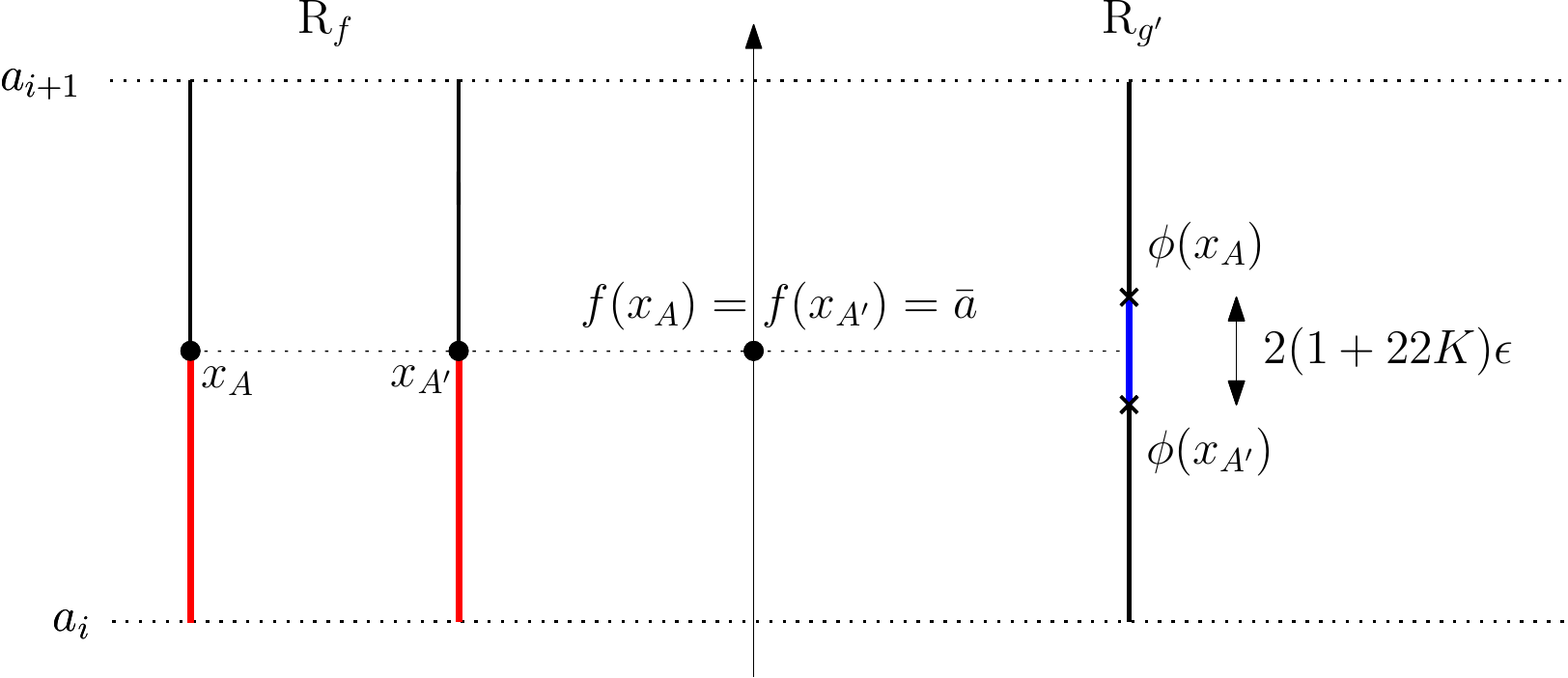}
\caption{\label{fig:smallBranches} Any path between $x_A$ and $x_{A'}$ must contain the red segments, and the blue segment is a
particular path between $\phi(x_A)$ and $\phi(x_{A'})$.}
\end{figure}

Hence, $D(\phi, \psi) \geq
|d_f(x_A,x_{A'})-d_{g'}(\phi(x_{A}),\phi(x_{A'}))|> a_f/2-2(1+22K)\e$,
which is greater than $2(1+22K)\e$ because $\e<a_f/(8(1+22K))$. Thus,
$\distfd(\Reeb_f,\Reeb_{g'})>(1+22K)\e$, which leads to a
contradiction with~(\ref{eq:tri}).  This means that there cannot be
more arcs in $f^{-1}((a_i,a_{i+1}))$ than in $(g')^{-1}((a_i,a_{i+1}))$.
Since the proof is symmetric in $f$ and $g'$, the numbers of arcs in
$(g')^{-1}((a_i,a_{i+1}))$ and in $f^{-1}((a_i,a_{i+1}))$ are actually the
same.

Finally, we show that the attaching maps of these arcs are also the
same.  In this particular graph setting, this is equivalent to showing
that corresponding arcs in $\Reeb_f$ and $\Reeb_{g'}$ have the same
endpoints.  Let $a_i$ be a critical value. Let $A_{f,i}^-$ and
$A_{f,i}^+$ (resp. $A_{g',i}^-$ and $A_{g',i}^+$) be the sets of arcs
in $f^{-1}((a_{i-1},a_i))$ and $f^{-1}((a_{i},a_{i+1}))$
(resp. $(g')^{-1}((a_{i-1},a_i))$ and $(g')^{-1}((a_{i},a_{i+1}))$).
Morevover, we let $\zeta_f^i$ and $\xi_f^i$ (resp. $\zeta_{g'}^i$ and
$\xi_{g'}^i$) be the corresponding attaching maps that send arcs to
their endpoints in $f^{-1}(a_i)$ (resp. $(g')^{-1}(a_i)$).  Let
$A,B\in A_{f,i}^-$. We define an equivalence relation $\sim_{f,i}$
between $A$ and $B$ by: $A\sim_{f,i} B$ iff $\zeta_f^i(A)=\zeta_f^i(B)$,
i.e. the endpoints of the arcs in the critical slice $f^{-1}(a_i)$ are
the same.  Similarly, $C,D\in A_{f,i}^+$ are equivalent if and only if
$\xi_f^i(C)=\xi_f^i(D)$.  One can define $\sim_{g',i}$ in the same way.
To show that the attaching maps of $\Reeb_f$ and $\Reeb_{g'}$ are the same,
we need to find a bijection $b$ between the arcs of $\Reeb_f$ and $\Reeb_{g'}$ 
such that $A\sim_{f,i}B \Leftrightarrow b(A)\sim_{g',i}b(B)$ for each $i$.

We will now define $b$ then check that it satisfies the condition.
Recall from~(\ref{eq:tri}) that 
$\distfd(\Reeb_f,\Reeb_{g'})< (1+22K)\e.$
Hence there exists a continuous map $\phi:\Reeb_f\rightarrow\Reeb_{g'}$ such that
$\|f-g'\circ\phi\|_\infty < (1+22K)\e$.
This map induces a bijection $b$ between the arcs of $\Reeb_f$ and $\Reeb_{g'}$.
Indeed, given an arc $A\in A_{f,i}^-$, let $x\in A$ such that $f(x)=\bar{a}=\frac{1}{2}(a_{i-1}+a_i)$. 
We define $b(A)$ as the arc of $A_{g,i}^-$ that contains $\phi(x)$.
The map $b$ is well-defined since 
$g'\circ\phi(x)\in\left[\bar a-(1+22K)\e,\bar a+(1+22K)\e\right]\subseteq(a_{i-1},a_i),$
hence $\phi(x)$	must belong to an arc of $(g')^{-1}((a_{i-1},a_i))$.
Let us show that  $b(A)\sim_{g',i}b(B)\Rightarrow A\sim_{f,i} B$.
Assume there exist $A,B\in A_{f,i}^-$ (the treatment of $A,B\in A_{f,i}^+$ is similar) such that $A\not\sim_{f,i} B$
and $b(A)\sim_{g',i} b(B)$. 
Let $x=\zeta_f^i(A)$ and $y=\zeta_f^i(B)$. Then we have $d_f(x,y)\geq a_f$ while $d_{g'}(\phi(x),\phi(y)) < 2(1+22K)\e$
(see Figure~\ref{fig:attach}).
Hence $|d_f(x,y)-d_{g'}(\phi(x),\phi(y))| > a_f - 2(1+22K)\e > 2(1+22K)\e$,
so $\distfd(\Reeb_f,\Reeb_{g'}) > (1+22K)\e$, which leads to a contradiction with~(\ref{eq:tri}).
The same argument applies to show that $ A\sim_{f,i} B\Rightarrow b(A)\sim_{g',i}b(B)$.
\end{proof}

\begin{figure}[htb]
\centering
\includegraphics[width=7.5cm]{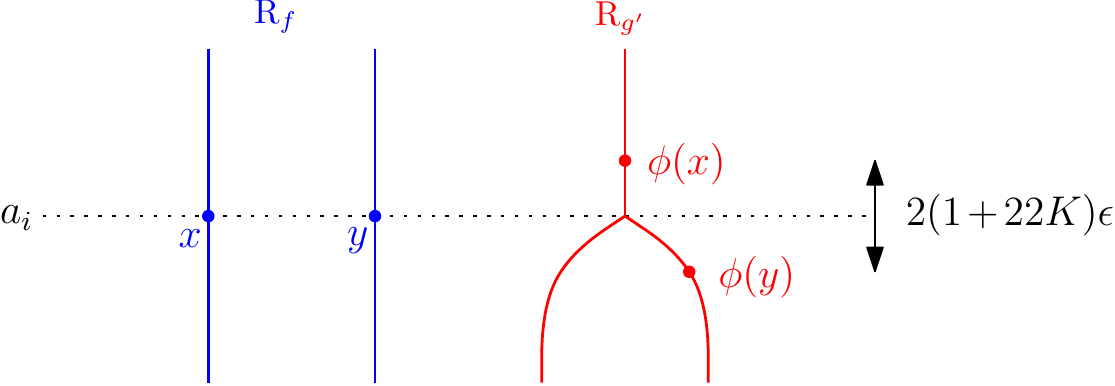}
\caption{Any path from $x$ to $y$ must go through an entire arc, hence $d_f(x,y)\geq a_f$.
On the contrary, there exists a direct path between $\phi(x)$ and $\phi(y)$, hence $d_{g'}(\phi(x),\phi(y)) < 2(1+22K)\e$.}
\label{fig:attach}
\end{figure}

\section{Induced Intrinsic Metrics}\label{sec:induced}

In this section we leverage the local equivalence given by
Theorem~\ref{th:locisom} to derive a global equivalence between the
intrinsic metrics $\hdistb$ and $\hdistfd$ induced by $\distb$ and
$\distfd$.  Note that we already know $\hdistfd$ to be equivalent to
$\hat d_{\rm GH}$ and $\hat d_{\rm I}$ since $\distfd$ is equivalent
to $d_{\rm GH}$ and $d_{\rm I}$. To the best of our knowledge, the
question whether $\distfd$, $d_{\rm I}$ or $d_{\rm GH}$ is intrinsic
on the space of Reeb graphs has not been settled, although $d_{\rm
  GH}$ itself is known to be intrinsic on the larger space of compact
metric spaces---see e.g.~\cite{Ivanov16}.

\subparagraph*{Convention.} In the following, whatever the metric  $d:\RS\times\RS\rightarrow\R_+$ under consideration, we define the class of {\em admissible paths} in
$\RS$ to be those maps $\gamma:[0,1]\rightarrow\RS$ that are
continuous in $\distfd$. This makes sense when $d$ is either $\distfd$
itself, $d_{\rm GH}$, or $d_{\rm I}$, all of which are equivalent to
$\distfd$ and therefore have the same continuous maps
$\gamma:[0,1]\rightarrow\RS$. In the case $d=\distb$ our convention
means restricting the class of admissible paths to a strict subset of
the maps $\gamma:[0,1]\rightarrow\RS$ that are continuous in~$d$ (by
Theorem~\ref{th:lowerbound}), which is required by some of our
following claims.

\begin{definition}
Let $d:\RS\times\RS\rightarrow\R_+$ be a metric on $\RS$. Let $\Reeb_f,\Reeb_g\in\RS$,
and $\gamma:[0,1]\rightarrow\RS$ be an admissible path such that $\gamma(0)=\Reeb_f$ and $\gamma(1)=\Reeb_g$.
The {\em length} of $\gamma$ induced by $d$ is defined as
$L_d(\gamma)= \sup_{n,\Sigma}\ \sum_{i=0}^{n-1} d(\gamma(t_{i}), \gamma(t_{i+1}))$
where $n$ ranges over~$\N$ and $\Sigma$ ranges over all partitions $0=t_0\leq t_1\leq ... \leq t_n=1$ of $[0,1]$.
The {\em intrinsic metric induced by $d$}, denoted $\hat d$, is defined by
$\hat d(\Reeb_f,\Reeb_g)=\rm{inf}_\gamma\ L_d(\gamma)$
where $\gamma$ ranges over all admissible paths $\gamma:[0,1]\rightarrow\RS$ 
such that $\gamma(0)=\Reeb_f$ and $\gamma(1)=\Reeb_g$.
\end{definition}

The following result is, in our view, the starting point for the study
of intrinsic metrics over the space of Reeb graphs. It comes as a
consequence of the (local or global) equivalences between $\distb$ and
$\distfd$ stated in Theorems~\ref{th:lowerbound} and~\ref{th:locisom}.
The intuition is that integrating two locally equivalent metrics along
the same path using sufficiently small integration steps
yields the same total length up to a constant factor, hence the
global equivalence between the induced intrinsic
metrics\footnote{Provided the induced metrics are defined using the
  same class of admissible paths, hence our convention.}.

\begin{thm}\label{th:strongeq}
$\hdistb$ and $\hdistfd$ are globally equivalent. Specifically, for any $\Reeb_f,\Reeb_g\in\RS$,
\begin{equation}\label{eq:strongeq}
\hdistfd(\Reeb_f,\Reeb_g)/\const\leq \hdistb(\Reeb_f,\Reeb_g)\leq 2\,\hdistfd(\Reeb_f,\Reeb_g).
\end{equation}
\end{thm}

\begin{proof}
We first show that $\hdistb(\Reeb_f,\Reeb_g)\leq 2\,\hdistfd(\Reeb_f,\Reeb_g)$. 
Let $\gamma$ be an admissible path 
and let $\Sigma=\{t_0,...,t_n\}$ be a partition of $[0,1]$.
Then, by Theorem~\ref{th:lowerbound}, 
$$\sum_{i=0}^{n-1} \distfd(\gamma(t_{i}), \gamma(t_{i+1}))\geq
\frac{1}{2}\sum_{i=0}^{n-1} \distb(\gamma(t_{i}), \gamma(t_{i+1})).$$
Since this is true for any partition $\Sigma$ of any finite size $n$, it
follows that
$$L_{\distfd} (\gamma) \geq \frac{1}{2} L_{\distb}(\gamma)\geq\frac{1}{2} \hdistb(\Reeb_f,\Reeb_g).$$
Again, this inequality holds for any admissible path $\gamma$, so $\hdistb(\Reeb_f,\Reeb_g)\leq 2\hdistfd(\Reeb_f,\Reeb_g)$. \\
We now show that $\hdistfd(\Reeb_f,\Reeb_g)/\const\leq \hdistb(\Reeb_f,\Reeb_g)$.
Let $\gamma$ be an admissible path 
and $\Sigma=\{t_0,...,t_n\}$ a partition of $[0,1]$.  We claim that
there is a refinement of $\Sigma$ (i.e. a partition
$\Sigma'=\{t'_0,...,t'_m\}\supseteq \Sigma$ for some $m\geq n$) such
that $\distfd(\gamma(t'_{j}),\gamma(t'_{j+1})) <
\max\{c_{t'_j},c_{t'_{j+1}}\}/16$ for all $j \in \{0,...,m-1\}$, where
$c_t>0$ denotes the minimal distance between consecutive
critical values of $\gamma(t)$. Indeed, since $\gamma$ is continuous
in $\distfd$, for any $t\in[0,1]$ there exists $\delta_t >0$ such that
$\distfd(\gamma(t),\gamma(t')) < c_t/16$ for all $t'\in [0,1]$ with
$|t-t'| < \delta_t$.  Consider the open cover
$\{(\max\{0,t-\delta_t/2\},\min\{1,t+\delta_t/2\})\}_{t\in[0,1]}$ of
$[0,1]$.  Since $[0,1]$ is compact, there exists a finite subcover
containing all the intervals $(t_i-\delta_{t_i}/2,t_i+\delta_{t_i}/2)$
for $t_i\in\Sigma$.  Assume w.l.o.g. that this subcover is minimal (if
it is not, then reduce the $\delta_{t_i}$ as much as needed).  Let then
$\Sigma'=\{t'_0,...,t'_m\}\supseteq \Sigma$ be the partition of
$[0,1]$ given by the midpoints of the intervals in this subcover,
sorted by increasing order.  Since the subcover is minimal, we have
$t'_{j+1}-t'_j < (\delta_{t'_j}+\delta_{t'_{j+1}})/2 <
\max\{\delta_{t'_j},\delta_{t'_{j+1}}\}$ hence
$\distfd(\gamma(t'_j),\gamma(t'_{j+1})) <
\max\{c_{t'_j},c_{t'_{j+1}}\}/16$ for each $j\in\{0,m-1\}$. It follows
that
\begin{align}
\sum_{i=0}^{n-1} \distfd(\gamma(t_{i}), \gamma(t_{i+1}))&\leq\sum_{j=0}^{m-1} \distfd(\gamma(t'_{j}), \gamma(t'_{j+1})){\rm\ by\ the\ triangle\ inequality\ since\ }\Sigma'\supseteq\Sigma\nonumber\\
&\leq \const\sum_{j=0}^{m-1} \distb(\gamma(t'_{j}), \gamma(t'_{j+1})){\rm\ by\ Theorem~\ref{th:locisom}\ with\ }K=1/22\nonumber\\
&\leq \const\, L_{\distb}(\gamma)\nonumber.
\end{align}
Since this is true for any partition $\Sigma$ of any finite size~$n$, it follows that 
$$\hdistfd(\Reeb_f,\Reeb_g)\leq L_{\distfd}(\gamma)\leq \const\, L_{\distb}(\gamma).$$
Again, this inequality is true for any admissible path $\gamma$, so $\hdistfd(\Reeb_f,\Reeb_g)\leq \const\,\hdistb(\Reeb_f,\Reeb_g)$.
\end{proof}

Theorem~\ref{th:strongeq} implies in particular that $\hdistb$ is a
true metric on Reeb graphs, as opposed to $\distb$ which is only a
pseudo-metric.  Moreover,
 the simplification operator defined in
Section~\ref{sec:smoothope} makes it possible to continuously deform any Reeb
graph into a trivial segment-shaped graph then into the empty graph. This shows that $\RS$ is
path-connected in $\distfd$. Since the length of such continuous deformations is finite
if the Reeb graph is finite, $\hdistfd$ and $\hdistb$ are finite
metrics. Finally, the global equivalence of $\hdistfd$ and
$\hdistb$ yields the following:

\begin{corollary}\label{cor:tcc} 
The metrics $\hdistfd$ and $\hdistb$ induce the same topology on $\RS$,
which is a refinement of the ones induced by $\distfd$ or $\distb$.
\end{corollary}

\begin{remark}
Note that the first inequality in~(\ref{eq:strongeq}) and,
consequently, Corollary~\ref{cor:tcc}, are wrong if one defines the
admissible paths for $\hdistb$ to be the whole class of maps
$[0,1]\to\RS$ that are continuous in $\distb$---hence our convention.
For instance, let us consider the two Reeb graphs $\Reeb_f$ and
$\Reeb_g$ of Figure~\ref{fig:cetrue} such that $\Dg(f)=\Dg(g)$, and
let us define $\gamma:[0,1]\rightarrow\RS$ by $\gamma(t)=\Reeb_f$ if
$t\in[0,1/2)$ and $\gamma(t)=\Reeb_g$ if $t\in[1/2,1]$.  Then $\gamma$
  is continuous in $\distb$ while it is not in $\distfd$ at $1/2$
  since $\distfd(\Reeb_f,\Reeb_g)>0$.  In this case,
  $\hdistb(\Reeb_f,\Reeb_g)\leq
  L_{\distb}(\gamma)=0<\hdistfd(\Reeb_f,\Reeb_g)$. 
\end{remark}

\section{Discussion}

In this article, we proved that the bottleneck distance, even though
it is only a pseudo-metric on Reeb graphs, can actually discriminate a
Reeb graph from the other Reeb graphs in a small enough neighborhood,
as efficiently as the other metrics do.  This theoretical result
legitimates the use of the bottleneck distance to discriminate between
Reeb graphs in applications. It also motivates the study of intrinsic
metrics, which can potentially shed new light on the structure of the
space of Reeb graphs and open the door to new applications where
interpolation plays a key part.  This work has raised numerous
questions, some of which we plan to investigate in the upcoming
months:
\begin{itemize}
\item {\bf Can the lower bound be improved?} We believe that $\e/22$
  is not optimal.  Specifically, a more careful analysis of the
  simplification operator should allow us to derive a tighter upper
  bound than the one in Lemma~\ref{lem:stabsmooth}, and
  improve the current lower bound on $\distb$.
\item {\bf Do shortest paths exist in $\RS$?} The existence of
  shortest paths achieving $\hdistb$ is an important question since a
  positive answer would enable us to define and study the {\em
    intrinsic curvature} of $\RS$.  Moreover, characterizing and
  computing these shortest paths would be useful for interpolating
  between Reeb graphs.  The existence of shortest paths is guaranteed
  if the space is complete and locally compact.  Note that $\RS$ is
  not complete, as shown by the counter-example of
  Figure~\ref{fig:ce}. Hence, we plan to restrict the focus to the
  subspace of Reeb graphs having at most $N$ features with height at
  most $H$, for fixed but arbitrary $N,H>0$. We believe this subspace
  is complete and locally compact, like its counterpart in the space
  of persistence diagrams~\cite{Blumberg14}.
  \begin{figure}[h]\centering
  \includegraphics[width=7.5cm]{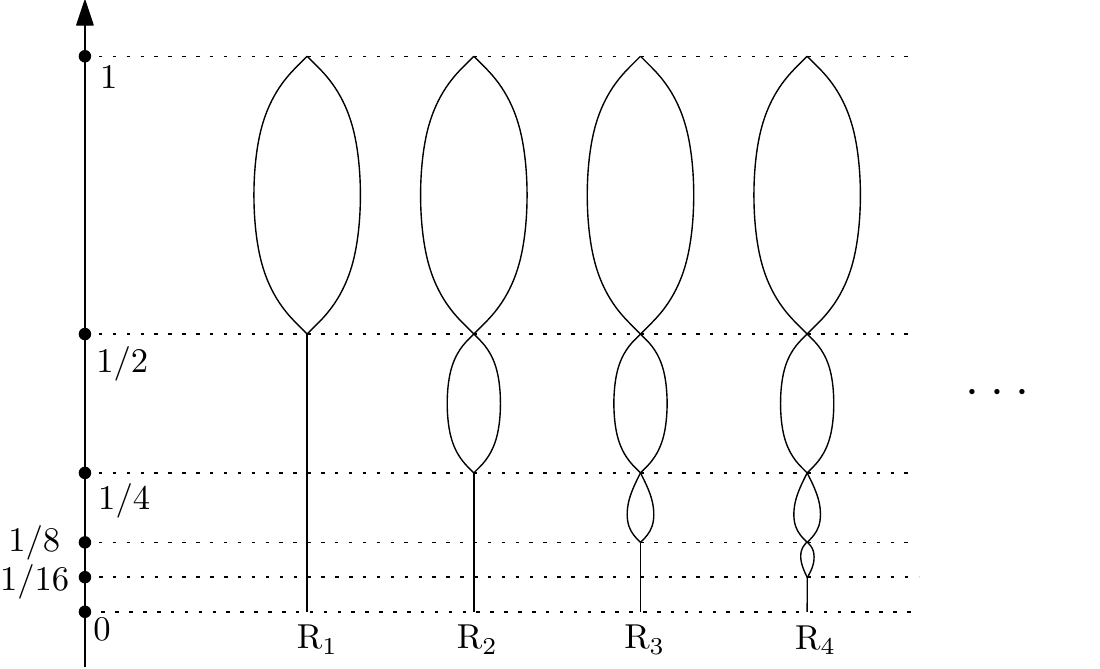}
  \caption{\label{fig:ce}
  A sequence of Reeb graphs that is Cauchy but that does not converge in $\RS$ because the number of critical values goes to $+\infty$.
  Indeed, each $\Reeb_n$ has $n+2$ critical values.}
  \end{figure}
\item {\bf Is $\RS$ an Alexandrov space?} Provided shortest paths
  exist in $\RS$ (or in some subspace thereof), we plan to determine whether
  the intrinsic curvature is bounded, either from above or from
  below.  This is interesting because barycenters in metric spaces
  with bounded curvature enjoy many useful
  properties~\cite{Ohta12}, and they can be approximated
  effectively~\cite{Ohta09}.
\item {\bf Can the local equivalence be extended to general metric
  spaces?} We have reasons to believe that our local equivalence
  result can be used to prove similar results for more general 
  classes of metric spaces than Reeb graphs. If true, this 
  would shed new light on inverse problems in persistence theory.
\end{itemize}  

\bibliographystyle{plain}
\bibliography{biblio}

\end{document}